\documentclass[a4paper,12pt]{article}

\usepackage[english]{babel}
\usepackage{color,float}
\usepackage{babel}
\usepackage{enumitem}
\usepackage{fontenc}
\usepackage{graphicx}
\usepackage{amsmath,amsthm,amssymb}
\usepackage[colorlinks=true,citecolor=blue]{hyperref}
\usepackage[usenames,dvipsnames]{pstricks}
\usepackage{amsmath, amsfonts, amssymb, mathrsfs,hyperref}
\usepackage{epsfig}\usepackage[square,sort,comma,numbers]{natbib}
\usepackage{natbib,hyperref}
\usepackage{amsthm}
\usepackage{subfigure}
\usepackage[utf8]{inputenc}
\usepackage[T1]{fontenc}
\usepackage{lmodern}
\usepackage{url}
\title{Zika Virus  Model}
\newtheorem{theorem}{Theorem}
\date{}
\title{The Impact of Mobility between Rural Areas and Forests on the Spread of Zika
}
\author{Kifah Al-Maqrashi,  
Fatma Al-Musalhi,
Ibrahim M. Elmojtaba, 
Nasser Al-Salti,\\
 Sultan Qaboos University, Muscat, Oman\\
 Email: fatma@squ.edu.om}

\begin{document}
\maketitle

\section*{Abstract}

A mathematical model of Zika virus transmission  incorporating human movement between rural areas and nearby forests is presented to investigate the role of human movement in the spread of Zika virus infections in human and mosquito  populations. Proportions of both susceptible and infected humans living in rural areas are assumed to move to nearby forest areas. Direct, indirect and vertical transmission routes are incorporated for all populations. Mathematical analysis of the proposed model has been presented. The analysis starts with normalizing the proposed model. Positivity and boundedness of solutions to the normalized model have been then addressed. The basic reproduction number has been calculated using the next generation matrix method and its relation to the three routes of disease transmission has been presented. The sensitivity analysis of the basic reproduction number to all model parameters has been investigated. The analysis also includes existence and stability of disease free and endemic equilibrium points. Bifurcation analysis has been also carried out. Finally, numerical solutions to the normalized model have been obtained to confirm the theoretical results and to demonstrate the impact of human movement in the disease transmission in human and mosquito populations.\\
\noindent {\bf Keywords:} Zika; vertical transmission; Basic Reproduction Number; Stability Analysis; Sensitivity Analysis; Bifurcation analysis.\\

%---- http://www.ams.org/msc/msc2010.html ---%
\noindent {\bf 2010 AMS Subject Classification:}  34C23,  34D23, 92D30, 93A30.

\section{Introduction}
Zika is an arboviral disease in genus flavivirus  closely related to yellow fever, West Nile (WN) and dengue (DEN) viruses. It is firstly identified in 1947 in Zika Forest in Uganda during sylvatic yellow fever surveillance in a sentinel rhesus monkey \cite{4}. In 1954, it is reported in human for the first time in Nigeria \cite{6}. Zika epidemic stated as a Public Health Emergency of Intentional Concern (PHEIC) by World Health Organizing (WHO) on first of February 2016 \cite{*}. It has attracted global attention since it has worldwide spread among tropical and subtropical regions. In Yap Island, Micronesia in 2007, the first Zika outbreak occurred among humans \cite{****}. During 2013-2014 the largest epidemic of Zika ever reported was in French Polynesia \cite{****}. Since 2014, Zika virus (ZIKV) has continued spreading to other pacific islands \cite{6}. It reached southern and Central America after 2015 and  Brazil and Caribbean were highly affected by ZIKV \cite{****}. Local transmission of ZIKV was realized in 34 countries by March 2016 \cite{19}.
\\
ZIKV is transmitted primarily to human population by bites of infected female Aedes mosquito. Analysts have found 19 species of Aedes mosquitoes competent of carrying Zika infection, but the foremost common is the tropical privateer, Aedes aegypti. The vector (mosquito) can pass human population through biting after taking a blood meal from infected human. In addition, sexual interaction, perinatal transmission and blood transfusion are other routes of spreading ZIKV between human even months after infection. A pregnant lady can pass Zika to her baby, which can cause genuine birth defects. Infection of Zika increased chances of developing the infants injury with microcephaly as which reported in \cite{5} and Guillian syndrome which reported in \cite{6} from infected mothers \cite{****}. In February 2016, France registered the first sexually transmitted case of ZIKV \cite{*}.
\\
 Zika disease is characterized by mild symptoms including fever, headache, maculopapular rash, joint and muscle pain and conjunctives, etc. The clinical symptoms duration is within two to seven days after the bites \cite{*}. Most reports show that Zika is a self – limiting febrile disease that could be misidentified as dengue or chikungunya fever \cite{7}.  
 \\
 The prevention of mosquitoes bites and control of vectors by using insecticide, eradication of adult and larval breeding areas is the only possible given treatment available till now \cite{**}.
 \\
  
Understanding the virus transmission and disease epidemiology through mathematical modeling is of great importance for disease management. A number of mathematical models have been developed to study the dynamics and propose control strategies for transmission of ZIKV disease. In \cite{*}, authors proposed Zika mathematical model by assuming the standard incidence type interaction of human to human transmission of the illness. Also, they extended their work to include optimal control programs (insecticide- treated bed nets, mosquito repulsive lotions and electronic devices) in order to reduce biting rate of vector to decline spreading of the disease among human population. In \cite{**},   authors proposed Zika mathematical model  including the applications of prevention, treatments and insecticide as a best way to minimize the spreading of ZIKV disease. In \cite{***}, researchers suggested multifold Zika mathematical model. They considered transmission of the ZIKV in the adult population and infants either directly by vector bites or through vertical transmition from mothers. The model shows that asymptomatic individuals magnify the disease weight in the community. Also, they explained that postponing conception, coupling aggressive vector control and personal protection use decrease the cases of microcephaly and transmission of Zika virus.
\\
\\
%Several studies concentrate on their works on the control of vector population, as mentioned before, whenever this vector live.
Globally, the survival of around 1.6 billion rustic people depends on products obtained from local forests, in whole or in part. Those individuals live  adjacent to the forest and they have had simple survival conditions and livelihoods for many generations. They depend on those natural and wild resources to meet their needs \cite{livelihoods}.
In this paper, a mathematical model of ZIKV is constructed to demonstrate the specific and realistic conditions, where nearby movement of humans may contribute to the spread of virus infections. This happens when an infected human with mild symptoms moves from rural areas to nearby forest areas, looking for jobs or food. Additionally, the movement of a susceptible human can affect the spreading of infections via contagious mosquitoes in the forest. Hence, in this paper, we have split the vector compartment, based on mosquito location, into rural areas and nearby forest alienated areas. Human movement between rural areas and their interaction with vector populations are illustrated in Figure \ref{f1}.
In addition,  sexual and vertical transmission in human population are considered. Also, vertical transmission from a contaminated female mosquito to their offspring is suggested as a component that guarantees upkeep of ZIKV.
\\The paper is organized as follows:  model formulation is described in Section 2. 
The model analysis including positivity, boundness of solution, basic reproduction number and  sensitivity analysis   are discussed in Section 3.  Furthermore, stability analysis and bifurcation analysis  are presented.
Numerical analysis of the model  using assumed baseline parameters are given in Section 4 to illustrate the   effects  of highly sensitive parameters on human population. Finally, conclusion is given in Section 5.

\section{Model Description}

In this section, we introduce a model for Zika virus transmission between
humans and vectors in rural areas and nearby forests.
We begin the description of the model with the human compartments. We split the human population into 
susceptibles $S_h$, symptomatic $I_h$ and recovered
$R_h$. Susceptible humans $ S_h$ can get infected with Zika via three main routes \cite{CDC1}: via a mosquito
bite (vector transmission), via sexual transmission or blood transfusion 
(direct tranmission), or by being passed from mother to a newborn child
(vertical transmission). Zika causes nearly no mortality among humans,
and has been a public health crisis for relatively
 short period of time, so we 
assume the total human population remains  constant (i.e) $ S_h+I_h+R_h=N_H $. 

We assume that a fraction $\epsilon_1$ of newborns to
infected symptomatic are affected and  that affected
 newborns enter the symptomatic class.
 Evidence suggests that  fraction is
about $2/3$ \cite{brasil2020zika}. We  also assume that susceptible humans can get
infected by infectious mosquitoes that live in rural areas $I_v$ and from
infectious humans (symptomatic ) through sexual transmission
or other direct routes. A proportion $\kappa_1$ of the susceptible individuals may also get infected  by infectious mosquitoes that live in forest alienated nearby rural areas $I_{u}$ due to their movement to forest areas. A proportion $\kappa_2$ of the  infected
individuals are also assumed to move from rural areas to the nearby forests such that $\kappa_1 > \kappa_2$ and hence they may infect mosquitoes that live in forests.  
We split the vector population into rural population ($S_v$, $I_v$) and nearby forest population ($S_{u}$, $I_{u}$) and because mosquitoes travel over distances of not more than a few kilometers, they will have a direct interaction only with vector and human populations in their corresponding areas.
Evidence suggests that Zika virus is transmitted vertically in the mosquito vector
\cite{Than} and this is the main pathway it survives the colder months.
We incorporate vertical transmission $ \epsilon_2,\epsilon_3 $ of the Zika virus in both vector populations, respectively.
Rural and forest mosquitoes are assumed to be only infected by infectious humans.
The infection period of mosquitoes ends when the mosquitoes die. The overall vectors population at time $ t  $ are $ S_v+I_v=N_V  $ and $ S_{u}+I_{u}=N_{U}. $
\begin{center}
\begin{figure}[H]
\includegraphics[scale=.45]{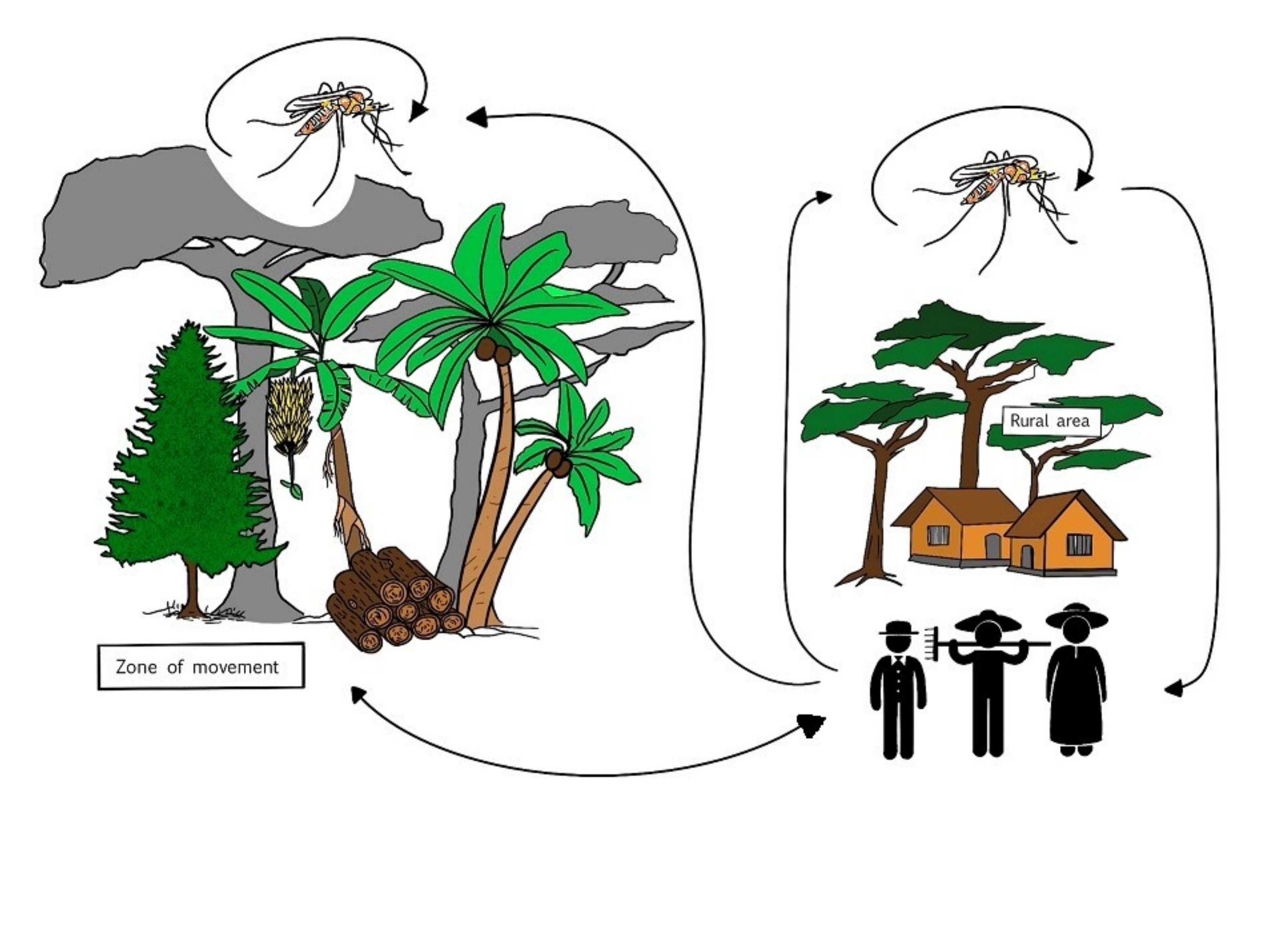}\hspace{0cm}
\caption{Illustrated figure for human movement and their interactions with vector populations of  ZIKV.}
\label{f1}
\end{figure}
\end{center}
%\vspace{0cm}
The set of  non-linear differential equations that represents  the proposed mathematical model is given by:

\begin{eqnarray}{\label{1}}
S_h^\prime & = & \mu_H N_H  - \mu_H  \epsilon_1 I_h  -  \beta_1 \theta_1 I_v \frac{S_h}{N_H} - \kappa_1\beta_2 \theta_1 I_{u} \frac{S_h}{N_H} - \lambda I_h  \frac{S_h}{N_H} - \mu_H S_h \nonumber\\
I_h^\prime & = & \mu_H \epsilon_1 I_h  +  \left(\beta_1 \theta_1 I_v \frac{S_h}{N_H} + \kappa_1\beta_2 \theta_1 I_{u} \frac{S_h}{N_H} + \lambda I_h  \frac{S_h}{N_H}\right)- (\gamma + \mu_H) I_h \nonumber\\
R_h^\prime & = &  \gamma I_h - \mu_H R_h  \nonumber\\
S_v^\prime & = &  \mu_V N_V - \mu_V  \epsilon_2 I_v - \beta_1 \theta_2 S_v \frac{I_h }{N_H}  - \mu_V S_v \label{model1} \\
I_v^\prime & = &  \mu_V \epsilon_2 I_v + \beta_1 \theta_2 S_v \frac{I_h }{N_H} - \mu_V I_v \nonumber\\
S_{u}^\prime & = &  \mu_U N_U - \mu_U  \epsilon_3 I_{u} - \beta_2 \theta_2 \kappa_2 S_{u}  \frac{ I_h}{N_H}  - \mu_U S_{u} \nonumber \\
I_{u}^\prime & = &  \mu_U \epsilon_3 I_{u} + \beta_2 \theta_2 \kappa_2 S_{u}  \frac{ I_h}{N_H}  - \mu_U I_{u} \nonumber
\end{eqnarray} 
with  non negative initial conditions $S_h(0), I_h(0),R_h(0),S_v(0),I_v(0),S_{u}(0),I_{u}(0).$ 
 \\In addition, the parameters of the system are defined in  Table \ref{table1}.
\begin{table}[H]
\caption{Parameters used in model (\ref{model1}) }
\label{table1}
\bigskip
\centering
\begin{tabular}{|l|c|}
%\toprule%
\hline
Parameter & Symbol \\
\hline
Natural death/birth rate of humans & $\mu_H$ \\
Natural death/birth rate of mosquito in rural areas & $\mu_V$ \\
Natural death/birth rate of mosquito in forest areas & $\mu_U$ \\
Biting rate of rural mosquitoes on humans & $\beta_1$ \\
Biting rate of forest mosquitoes on humans & $\beta_2$ \\
Transmission probability from an infectious mosquito
to a susceptible human  & $\theta_1$ \\
Transmission probability from an infectious human to
a susceptible mosquito  & $\theta_2$ \\
Direct (sexual) transmission rate between humans & $\lambda$\\

Recovery rate of humans & $\gamma$\\

Probability of vertical transmission in humans, rural mosquito, forest mosquitoes & $\epsilon_1$, $\epsilon_2$, $\epsilon_3$\\

Fraction of susceptible humans moving from rural to forest areas
&$\kappa_1$\\

Fraction of infected humans moving from rural to forest areas
& $\kappa_2$\\ 
\hline
\end{tabular}
\end{table} 
Let $ S_H= \dfrac{S_h}{N_H} ,$ $\, I_H= \dfrac{I_h}{N_H}, $  $\, R_H= \dfrac{R_h}{N_H}, $  $ \,S_V= \dfrac{S_v}{N_V}, $ $\, I_V= \dfrac{I_v}{N_V}, $ $\, S_U= \dfrac{S_{u}}{N_{U}}, $ $ \,I_U= \dfrac{I_{u}}{N_{U}}, $  \\such that $$S_H+I_H+R_H=1, \quad S_V+I_V=1, \quad S_U+I_U=1.  $$ 
\vspace{0.1cm}
\\
Thus, the considered model (\ref{1}) have been normalized and rewritten as follows: 

\begin{eqnarray}{\label{2}}
S_H^\prime & = & \mu_H   - \mu_H  \epsilon_1 I_H  -  \beta_1 \theta_1 \alpha_1 I_V S_H  -\kappa_1 \beta_2 \theta_1 \alpha_2 I_U S_H  - \lambda I_H S_H- \mu_H S_H \nonumber\\
I_H^\prime & = & \mu_H \epsilon_1 I_H  +  \left(\beta_1 \theta_1 \alpha_1 I_V S_H + \kappa_1 \beta_2 \theta_1 \alpha_2 I_U S_H  + \lambda I_H   S_H\right)- (\gamma + \mu_H) I_H \nonumber\\
R_H^\prime & = &  \gamma I_H - \mu_H R_H  \nonumber\\
S_V^\prime & = &  \mu_V  - \mu_V  \epsilon_2 I_V - \beta_1 \theta_2 S_V I_H   - \mu_V S_V \label{model2} \\
I_V^\prime & = &  \mu_V \epsilon_2 I_V + \beta_1 \theta_2 S_V I_H  - \mu_V I_V \nonumber\\
S_U^\prime & = &  \mu_U - \mu_U  \epsilon_3 I_U - \beta_2 \theta_2 \kappa_2 S_U I_H - \mu_U S_U \nonumber \\
I_U^\prime & = &  \mu_U \epsilon_3 I_U + \beta_2 \theta_2 \kappa_2 S_U I_H - \mu_U I_U \nonumber
\end{eqnarray} 
\\
 where   $ \alpha_1=\dfrac{N_V}{N_H} $ and $ \alpha_2=\dfrac{N_{U}}{N_H} $ and  with  non-negative initial conditions
  $$X(0):=(S_H(0),I_H(0),R_H(0),S_V(0),I_V(0),S_U(0),I_U(0))^T  .$$

\begin{center}
\begin{figure}[H]
\includegraphics[scale=.42]{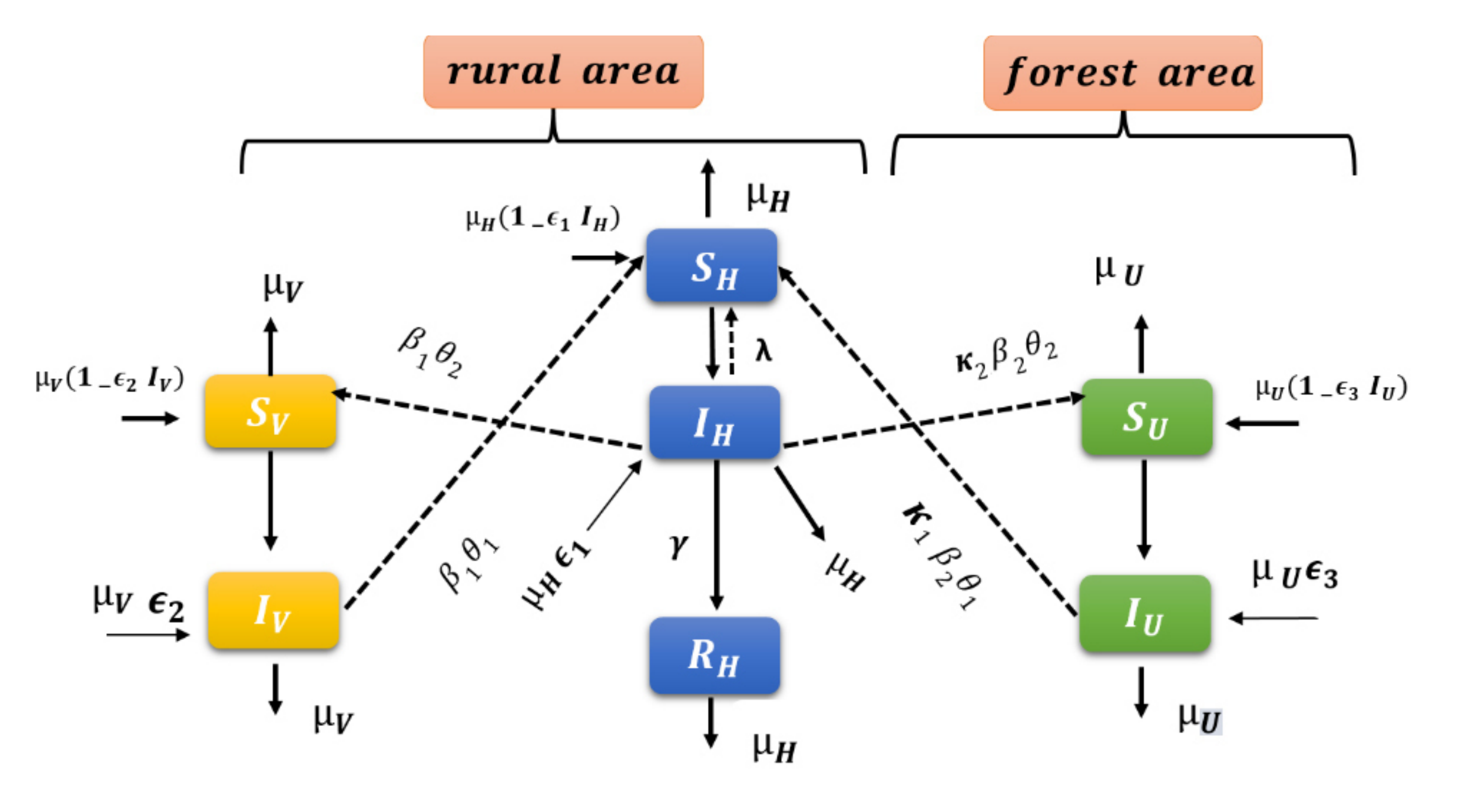}\hspace{.3cm}
\caption{Progression diagram of the proposed ZIKV model.  }
\end{figure}
\end{center}

%%%%%%%%%%%%%%%%%%%%%%%%%%%%%%%%%%%%%%%%%%%%%%%%%%%%%%

\section{ Model Analysis.}
In this section, the positivity of solutions, positive invariant set and the basic reproduction number have been  discussed. Also,  sensitivity analysis and results related to  stability analysis and bifurcation analysis are presented.
\subsection{ Positivity of Solutions and Positively Invariant Set}
It is clear that model (\ref{1}) together with initial conditions has a unique solution. Next we need to show that all solutions  remain non-negative for all $ t\in [0,\infty) $ for arbitrary choice of initial conditions in order to have an epidemiological convince results. The following theorem demonstrates the positivity and boundedness of state variables: 
\begin{theorem}
	The solutions  $S_H(t),\,I\,_H(t),\,R_H(t),\,S_V(t),\,\,I_V(t),\,S_U(t),$ and $I_U(t) $ of system (\ref{2}) with non-negative initial conditions $ S_H(0),I_H(0),R_H(0),S_V(0),I_V(0),S_U(0),I_U(0)$ remain positive for all time $t >0 $ in a positively invariant closed set $$\Omega := \left\lbrace (S_H,I_H,R_H,S_V,I_V,S_U, I_U)^T \in R^7_+:0\leqslant S_H(t),I_H(t),R_H(t),S_V(t),I_V(t),S_U(t),I_U(t)\leqslant 1      \right\rbrace. $$
\end{theorem}

\begin{proof}
	Assume that the initial conditions of system (\ref{2}) are non-negative. 
	Let $t_1>0$ be the first time at which there exists at least one component which is equal to zero and other components are non-negative on $[0,t_1]$. In the following, we will show that none of the components can be zero at $t_1$. Let us first, assume that $S_H(t_1)=0$ and other components are non-negative on $[0,t_1].$ Now, $S'_H$ can be written as 
	$$ S_H^\prime  = \mu_H (1-   \epsilon_1)+ \mu_H \epsilon_1 R_{H}- m_1 S_H-\mu_H (1-\epsilon_1)S_{H}, $$  where,  $	m_1 = \beta_1 \theta_1 \alpha_1 I_V +\kappa_1\beta_2 \theta_1 \alpha_2 I_U + \lambda I_H>0 .$ Then, at $t_1$, we have
	$$ \dfrac{S_H(t)}{dt}\bigg|_{t=t_1}  =  \mu_H (1-   \epsilon_1)+ \mu_H \epsilon_1 R_{H}(t_1)>0,  $$
	which means that $S_H(t)$ is strictly monotonically
	increasing at $t_1,$ i.e;  $S_H(t)<S_H(t_1)$ for all $t\in (t_1-\epsilon, t_1),$ where $\epsilon>0.$  Since $S_H(t_1)=0,$ then, $S_H(t)<0$ on $(t_1-\epsilon, t_1).$
	This leads to  a contradiction. Therefore, $S_H(t)$ can not be zero at $t_1$. Now, we assume that $I_H(t_1)=0$  and other components are non-negative. Then,
	$$ \dfrac{I_H(t)}{dt}\bigg|_{t=t_1} =(\beta_1 \theta_1 \alpha_1 I_V(t_1) +\kappa_1\beta_2 \theta_1 \alpha_2 I_U(t_1))S_H(t_1)>0,$$
	which means that $I_H(t)$ is strictly monotonically
	increasing at $t_1.$  Hence, we also get a contradiction.\\
	Next,  assume that $R_H(t_1)=0$  and other components are non-negative. Then, 
	$$\dfrac{R_H(t)}{dt}\bigg|_{t=t_1} =\gamma I_H(t_1)>0,$$
	which again leads to a contradiction.  \\
	In a similar manner, one can prove that the remaining components of vector populations $ S_V(t),I_V(t),S_U(t),I_U(t)  $ can not be zero at $t_1$.
	Hence, from the above, we conclude that such a point $t_1$ at which at least one component is zero does not exist.  Hence, all components remain positive for all time $t>0$.
	\\For the positively invariant closed set $\Omega$, we first note that the set $ \Omega $ is said to be positively invariant set if the initial conditions are in
	$  \Omega $ implies that $$ (S_H(t),I_H(t),R_H(t),S_V(t),I_V(t),S_U(t),I_U(t))^T \in\Omega. $$\\ Let 
	$$  \Phi(t)=(\Phi_1(t),\Phi_2(t),\Phi_3(t))^T                              $$
	where $\Phi_1(t)= S_H(t)+I_H(t)+R_H(t),\Phi_2(t)=S_V(t)+I_V(t),\Phi_3(t)= S_U(t)+I_U(t), $
	\\then 
	\[ \Phi'(t)= \left[ {\begin{array}{c}
			\mu_H-\mu_H\Phi_1(t)
			\\
			\mu_V-\mu_V\Phi_2(t)
			\\
			\mu_U-\mu_U\Phi_3(t)
	\end{array}}\right]. \] 
	
	\vspace{0.3cm}
	Now, solving for $ \Phi_1, \Phi_2$ and $\Phi_3$, we get 
		$$   \Phi_1(t)= 1-\left( 1-  \Phi_1(0)\right)e^{-\mu_H t},       $$
		$$   \Phi_2(t)= 1-\left( 1-  \Phi_2(0)\right)e^{-\mu_V t},       $$
	$$   \Phi_3(t)= 1-\left( 1-  \Phi_3(0)\right)e^{-\mu_U t},      $$
	where 
	$ \Phi_1(0)= S_H(0)+I_H(0)+R_H(0),\Phi_2(0)=S_V(0)+I_V(0)  $ and $ \Phi_3(0)= S_U(0)+I_U(0). $
	\\
	It is straight forward to conclude that $$ \Phi_1(t)\leqslant 1  \quad \text{if} \quad \Phi_1(0)\leqslant 1, $$
	$$ \Phi_2(t)\leqslant 1  \quad \text{if} \quad \Phi_2(0)\leqslant 1, $$
	$$ \Phi_3(t)\leqslant 1  \quad \text{if} \quad \Phi_3(0)\leqslant 1. $$
	Thus, we have  $ 0 \leqslant S_H(t),I_H(t),R_H(t),S_V(t),I_V(t),S_U(t),I_U(t) \leqslant 1 $ and hence the set  $ \Omega $ is  positively invariant set. Moreover, the set $ \Omega $ is a globally attractive set since  if $ \Phi_i(0)> 1$ then $  \underset{t \to \infty} {\lim}\Phi_i(t) =1 $ for $i=1,2,3.$
\end{proof}
%%%%%%%%%%%%%%%%%%%%%%%%%%%%%%%%%%%%%%%%%%%%%%%%%%%%%%%

\subsection{The Basic Reproduction Number }
The model (\ref{2}) has a disease free equilibrium (DFE)  :$$ Z^0 := (S_H^0,0,0,S_V^0,0,S_U^0,0)\in \Omega $$
where,  $S_H^0= S_V^0=S_U^0 =1 $. 
\\
The number of new infections produced by a typical infected individual in a population of DFE is called the basic reproduction number $ R_0 $ which can be obtained
by applying Next Generation Method \cite{ngm}. The  next generation matrix is
\\
\[PQ^{-1}= \left[ {\begin{array}{ccc}
\dfrac{\lambda}{-\mu_H \epsilon_1+\gamma+\mu_H} & \dfrac{\alpha_1\beta_1 \theta_1}{-\mu_V\epsilon_2+\mu_V} & \dfrac{\kappa_1\alpha_2 \beta_2 \theta_1}{ -\mu_U \epsilon_3+\mu_U}
\\
 \dfrac{\beta_1 \theta_2}{-\mu_H \epsilon_1+\gamma+\mu_H} & 0& 0
\\
\dfrac{\kappa_2 \beta_2 \theta_2}{-\mu_H \epsilon_1+\gamma+\mu_H} & 0 & 0
\end{array}}\right]. \]

\vspace{0.5cm}
where,  $ P $  is  Jacobian of the transmission matrix  describes the  production of new infections, whereas  $ Q $  is Jacobian of the transition matrix describes changes in state,  which are given by  
\[P = \left[ {\begin{array}{ccc}
\lambda & \alpha_1\beta_1 \theta_1 & \kappa_1\alpha_2 \beta_2 \theta_1
\\
 \beta_1 \theta_2 & 0 & 0
\\
\kappa_2 \beta_2 \theta_2 & 0 & 0
\end{array}}\right] \]
\\
 and 
 \\
 \[Q= \left[ {\begin{array}{ccc}
-\mu_H \epsilon_1+\gamma+\mu_H & 0 & 0
\\
 0 & -\mu_V\epsilon_2+\mu_V & 0
\\
0 & 0 & -\mu_U \epsilon_3+\mu_U
\end{array}}\right]. \]
\vspace{0.5cm}
 The basic reproduction number $ R_0 $ is the the dominant eigenvalue of $ PQ^{-1} $, which can be expressed as:
 \begin{equation}\label{10}
 R_0 = \dfrac{1}{2}\left[ R_{HH}+\sqrt{R_{HH}^2+ 4(R_{HV}+R_{HU})}\right]  
 \end{equation}
 
 where, 
 
 $\quad R_{HH}= \dfrac{\lambda}{\gamma+\mu_H(1-\epsilon_1)}  $,$\quad \quad R_{HV}= \dfrac{\alpha_1 \beta_1^2 \theta_1\theta_2 }{\mu_V (1-\epsilon_2)[\gamma+\mu_H(1-\epsilon_1)]} $ 
 \\ 
 \\
 and   $\quad R_{HU}= \dfrac{\kappa_1\kappa_2 \alpha_2 \beta_2^2 \theta_1\theta_2 }{\mu_U (1-\epsilon_3)[\gamma+\mu_H(1-\epsilon_1)]}. $ \\
 \\
 Note that $ R_{HH} $ denotes the reproduction number due to human to human transmission, $ R_{HV} $ denotes the reproduction number due to interaction between human and vector in rural area  and $ R_{HU} $ denotes the  reproduction number due to  interaction between human and vector in forest area. The square root represents the geometric mean, which means that  two steps are required for transmission of the disease. Moreover, the  threshold of the disease occurs at $ R_0=1 \Leftrightarrow  R_{HH}+R_{HV}+R_{HU}=1. $ Also, it can be easily proven that $ R_0<1 $ implies $ R_{HH}+R_{HV}+R_{HU}<1$, which means in order for the disease to die out, all the transmission roustes represented by $ R_{HH},$  $ R_{HV} $ and $ R_{HU} $ need to be reduced. Clearly, this will also imply that $R_0 > 1$ whenever $R_{HH}$, $R_{HV}$ or $R_{HU}$ is greater that one. 
 
%%%%%%%%%%%%%%%%%%%%%%%%%%%%%%%%%%%%%%%%%%%%%%%%%%
\subsection{Sensitivity Analysis of the Basic Reproduction Number.}

A fundamental and valuable numeric value for  the study   of infectious diseases dynamics is the basic reproduction number $ R_0,$ since it predicts whether an outbreak will expected to continue (when $R_0 > 1$) or die out (when $R_0 < 1$). Sensitivity analysis of the basic reproduction number allows us to determine which model parameters have the most impact on $R_0$.  A high   sensitive parameter leads to a high quantitative variation in 
$ R_0 $. Moreover, sensitivity analysis highlights the parameters that must be attacked by intervention and treatment strategies. Here, we adopt the elasticity index (normalized forward sensitivity index )\cite{nsens}, $E_P^{R_0}$, which computes the relative change of $ R_0 $ with respect to any parameter  $ P $ as follows
\begin{equation}\label{12}
 E_P^{R_0}= \dfrac{P}{R_0}\lim_{\bigtriangleup p \to 0}\dfrac{\bigtriangleup R_0}{\bigtriangleup P} = \dfrac{P}{R_0}\dfrac{\partial R_0}{\partial P}.
 \end{equation}
\vspace{0.1cm} 
Using the parameter values listed in Table \ref{table2} and the explicit expression of the basic reproduction number (\ref{10}), the estimated values of the elasticity indices and their interpretation are given in Table \ref{table3}. 
 
\begin{table}[H]
\caption{Parameters values used in model (\ref{2}) }
\label{table2}
\bigskip
\centering
\begin{tabular}{|c|c|c|}
%\toprule%
\hline
Parameter & Value &Source\\
\hline
$ \beta_1 $ & 0.3-1.5 bites per day per mosquito & \cite{biting}\\
$ \beta_2 $ &0.3-1.5 bites per day per mosquito & \cite{biting} \\
$ \lambda $ &0.27 &Assumed \\
$ \kappa_1 $ &0.5 & Assumed \\
$ \kappa_2 $ &0.3 & Assumed \\
$ \theta _1 $ & 0.1–0.75 per bites & \cite{biting} \\
$ \theta_2 $ & 0.3–0.75 per bites & \cite{biting} \\
$ \gamma $ & 1/6 per days& \cite{gamma} \\
$ \alpha_1 $ & 2& Assumed \\

$ \alpha_2 $ &3 &Assumed \\

$ \mu_H $ & 1/(60*365) per day & Assumed\\ 
$ \mu_V $ &1/14 per day &  \cite{vdeath}\\
$ \mu_U $ &1/14 per day &\cite{vdeath} \\
$ \epsilon_1 $ &0.67 & \cite{brasil2020zika} \\
$ \epsilon_2 $ & 0.06&\cite{lai2020vertical} \\
$ \epsilon_3 $ &0.06 &\cite{lai2020vertical}  \\
\hline
\end{tabular}
\end{table} 
\vspace{0.5cm}

\begin{table}[H]
\caption{Sensitivity indices and their interpretation  }
\label{table3}
\bigskip
\centering
\begin{tabular}{|c|c||c|}
%\topUule%
\hline
Parameter& Sensitivity index &Interpretion(increase or decrease)\\
\hline
$ \beta_1 $ & 0.5151 & $ \beta_1 $ by $10\%$ , $R_0 $ by $ 51\% $  \\
$\beta_2 $ & 0.0645&  $ \beta_2 $ by $10\%$ , $R_0 $ by $ 6.4\% $  \\
$ \lambda $ &0.4204 &  $ \lambda $ by $10\%$ , $R_0 $ by $ 42\% $ \\
$ \kappa_1 $ & 0.0323 &  $ \kappa_1 $ by $10\%$ , $R_0 $ by $ 3.2\% $ \\
$ \kappa_2 $ & 0.0323 &  $ \kappa_2 $ by $10\%$ , $R_0 $ by $ 3.2\% $ \\
$ \theta _1 $ &0.2898 &  $ \theta_1 $ by $10\%$ , $R_0 $ by $ 29\% $  \\
$ \theta_2 $ & 0.2898 &   $ \theta_2 $ by $10\%$ , $R_0 $ by $ 29\% $ \\
$ \gamma $ & -0.7099&  $ \gamma $ by $10\%$ , $R_0 $ by $ 71\% $  \\
$ \alpha_1 $ & 0.2575& $ \alpha_1 $ by $10\%$ , $R_0 $ by $ 26\% $  \\

$ \alpha_2 $ & 0.0323& $ \alpha_2 $ by $10\%$ , $R_0 $ by $ 3.2\% $  \\
$ \mu_H $ & -0.000196 & $ \mu_H $ by $10\%$ , $R_0 $ by $ 0.01\% $  \\ 
$ \mu_V $ & -0.2575 &  $ \mu_V $ by $10\%$ , $R_0 $ by $ 16\% $  \\
$ \mu_U $ &-0.0323 &  $ \mu_U $ by $10\%$ , $R_0 $ by $ 3.2\% $ \\
$ \epsilon_1 $ & 0.399e-5 & $ \epsilon_1 $ by $10\%$ , $R_0 $ by $ 0.0003\% $  \\
$ \epsilon_2 $ & 0.01073& $ \epsilon_2$ by $10\%$ , $R_0 $ by $ 1.07\% $  \\
$ \epsilon_3 $ & 0.000997 & $ \epsilon_3$ by $10\%$ , $R_0 $ by $ 0.09\% $  \\
\hline
\end{tabular}
\end{table}  
 \vspace{0.3cm}

Obviously, the corresponding model parameters will affect the basic reproduction number either positively or negatively. The positive sign of the sensitivity indices of the parameters $ \beta_1, \beta_2, \lambda , ,\kappa_1,
\kappa_2, \theta_1, \theta_2, \alpha_1 $ and $ \alpha_2 $ denotes the increase of the basic reproduction number $ R_0 $ as that parameter changes,  whereas the  negative sign of  the sensitivity indices of the parameters $ \gamma, \mu_H, \mu_V$ and $ \mu_U$ denotes the decrease of the basic reproduction number $ R_0 $ as that parameter changes. Moreover, the magnitude denotes the relative importance of the spotlight parameter. Clearly, the most efficacious parameter is the biting rate of rural mosquitoes on humans $\beta_1 $, i.e.,  it has a strong positive impact on the value of $ R_0. $  Also, the direct (sexual) transmission rate between humans  $  \lambda$   has  strong positive impact on $ R_0$.
 On the other hand,  the recovery rate of humans $ \gamma $ has the most negative sensitive index, it will decrease $ R_0 $ by $ 71\% $ when it increases by $ 10\%. $ The transmission probabilities per bite per human and  vector $ \theta_1  $ and $ \theta_2, $ respectively, have positive influence on the value of  $ R_0. $ Clearly, there is a very small  positive effect of  the vertical transmission of human and both vectors $ \epsilon_1,\epsilon_2,\epsilon_3 $  because their
elasticity indices are very small. Similarly, one can note that the proportions of movement for susceptible and infected humans $\kappa_1$ and $\kappa_2$ have a small positive effect on $R_0.$

%%%%%%%%%%%%%%%%%%%%%%%%%%%%%%%%%%%%%%%%%%%%%%%%%%

\subsection{Local Stability of the DFE}
Here we discuss the local stability of the DFE by finding the eigenvalues of  linearized system. The following theorem is devoted to the local stability of the DFE, i.e., the disease would be eliminated on certain time under certain conditions. 

 \begin{theorem}
 \label{ls}
 If $  R_0\leq  1, $  the  DFE of the model (\ref{2}) is locally asymptotically stable. If $  R_0> 1, $ it is unstable.
 \end{theorem} 
\vspace{0.3cm}

 \begin{proof}
 
  The linearized matrix of  the system (\ref{2}) at the disease free equilibrium $  Z^0 $ is
\[J_{Z^0} = \left[ {\begin{array}{ccccccc}
-\mu_H & -\mu_H \epsilon_1-\lambda & 0 & 0 &-\alpha_1\beta_1\theta_1 & 0 &-\kappa_1\alpha_2\beta_2\theta_1 
\\
0 & \mu_H \epsilon_1-\lambda -\gamma-\mu_H & 0 & 0 &\alpha_1\beta_1\theta_1 & 0 & \kappa_1\alpha_2\beta_2\theta_1
\\
0 & \gamma &  -\mu_H & 0 & 0 &0 & 0
\\
0 & -\beta_1\theta_2 & 0 & -\mu_V & -\mu _V\epsilon_2 &0 &0
\\ 
 0 & \beta_1\theta_2&0&0 &\mu_V\epsilon_2-\mu_V & 0&0
\\
0&-\kappa_2\beta_2\theta_2 &0&0&0&-\mu_U & -\mu_U \epsilon_3
\\
0 & \kappa_2\beta_2\theta_2 & 0 &0&0&0 &\mu_U \epsilon_3-\mu_U
\end{array}}\right]. \]
 It is clear that, the system has  four negative eigenvalues $\lambda_{1,2,3,4}=-\mu_H,-\mu_H,-\mu_V,-\mu_U.$
 The remaining eigenvalues can be found from the characteristic equation  $k(\lambda)=0,$ where the  $k(\lambda)$ given by
\begin{equation}\label{ce}
 k(\lambda)= \lambda^3+ k_1 \lambda^2+k_2 \lambda +k_3,  
 \end{equation}
with
$$\begin{array}{ll}
 k_1&= \mu_U(1-\epsilon_3)+\mu_V(1-\epsilon_2)+ \left[ \mu_H(1-\epsilon_1)+\gamma\right] (1-R_{HH}),\\
 k_2 &= [\gamma+\mu_H(1-\epsilon_1)]\lbrace \mu_U(1-\epsilon_3)(1-[R_{HH}+R_{HU}])+\mu_V(1-\epsilon_2)(1-[R_{HH}+R_{HV}])\rbrace \\&+\mu_V\mu_U(1-\epsilon_2)(1-\epsilon_3), \\
 k_3& =\mu_V\mu_U(1-\epsilon_2)(1-\epsilon_3)[\gamma+\mu_H(1-\epsilon_1)](1-[R_{HH}+R_{HV}+R_{HU}]).  \end{array}$$

It is clear that $k_3>0$ if $R_{HH}+R_{HV}+R_{HU}<1$ which also implies  that $k_1>0$ and $k_2>0.$
Hence, in order to use Routh's stability criterion \cite{routh} to show that the roots of the above characteristic equation have negative real part, it is left to show the $k_1k_2-k_3$ is positive. This can be shown as follows: \\
 $
 k_1k_2-k_3=$\\
 $\mu_U^2(1-\epsilon_3)^2[\gamma+\mu_H(1-\epsilon_1)][1-(R_{HH}+R_{HU})]+\mu_V^2(1-\epsilon_2)^2 [\gamma+\mu_H(1-\epsilon_1)][1-(R_{HH}+R_{HV})]+\mu_U(1-\epsilon_3)[\gamma+\mu_H(1-\epsilon_1)]^2(1-R_{HH})(1-R_{HU})+\mu_V(1-\epsilon_2)[\gamma+\mu_H(1-\epsilon_1)]^2(1-R_{HH})[1-(R_{HH}+R_{HV})]+2\mu_V\mu_U(1-\epsilon_2)(1-\epsilon_3)[\gamma+\mu_H(1-\epsilon_1)](1-R_{HH})+\mu_U^2\mu_V(1-\epsilon_2)(1-\epsilon_3)^2+\mu_U\mu_V^2(1-\epsilon_2)^2(1-\epsilon_3).  
$\\
 Clearly, $k_1k_2-k_3 >0$  if and only if  $\quad R_{HH}+R_{HV}<1  \text{ and } R_{HH}+R_{HU}<1. $ 
 \\
Therefore, by Routh's stability criterion, the roots of  the characteristic equation $k(\lambda)=0 $ have negative real part, and hence we conclude that the  DFE is locally asymptotically stable whenever $ R_0\leq 1. $ Otherwise, it is unstable.

\end{proof}

\subsection{Global Stability of the DFE}
 When the solution of the dynamical system (\ref{2}) approaches a unique equilibrium point regardless of initial conditions then the equilibrium point is globally asymptotically stable. The global stability of the DFE will  ensure that the disease is eliminated under all initial conditions. In this regard, we state the following theorem:
\begin{theorem}
\label{gs}
 If $ R_0\leq 1, $ the disease free equilibrium $ Z^0 $ is globally asymptotically stable  on the compact set $ \Omega. $
 \end{theorem}

\begin{proof}

 Applying Castillo-Chavez theorem \cite{gas}, consider the following two compartments:

 \[ X(t)= \left[ {\begin{array}{c}
S_H(t)
\\
R_H(t)
\\
 S_V(t)
 \\
 S_U(t)
\end{array}}\right], \qquad
 Y(t)= \left[ {\begin{array}{c}
I_H(t)
\\
 I_V(t)
 \\
 I_U(t)
\end{array}}\right] \]
which describe the uninfected and infected individuals of  system (\ref{2}), respectively.
 So that system (\ref{2}) can be written as
 
 $$ \dfrac{dX}{dt}= F(X,Y),\quad \quad \dfrac{dY}{dt}= G(X,Y);\hspace{0.3cm} G(X,0)=0,     $$
 \\
where $F(X,Y)$ and $G(X,Y)$ are the corresponding right hand side of system (\ref{2}). To guarantee the global asymptotic stability of the DFE,  according to Castillo-Chavez theorem, the following two conditions must be satisfied:
\\
\\{\bf{$(H1)$}} \hspace{0.1cm}  For $\dfrac{dX}{dt}= F(X,0),$ $ X^0= (1,0,1,1)^T $  is globally asymptotically stable.
\\{\bf{$(H2)$}} $\hat{G}\geqslant 0,$ where $ \hat{G}(X,Y)= AY - G(X,Y) $ and  $ A=D_YG(X^0,0) $ is an Metzler matrix $ \forall  (X,Y)\in \Omega. $
\\
To check the first condition, we find
 \[F(X,0)= \left[ {\begin{array}{c}
-\mu_H S_H+\mu_H 
\\
-\mu_H R_H
\\
 -\mu_V S_V+\mu_V
\\
-\mu_U S_U+\mu_U
\end{array}}\right]. \]
Solving the system of  ODEs in $(H1)$, we  obtain the following behavior of each component:

$$ S_H(t) = 1+  S_H(0)e^{-\mu_H t} \Rightarrow \lim_{t \to \infty}S_H(t)=1,  $$
$$ R_H(t) = R_H(0) e^{-\mu_H t} \Rightarrow \lim_{t \to \infty}R_H(t)=0,  $$ 
$$ S_V(t) = 1+S_V(0) e^{-\mu_V t} \Rightarrow \lim_{t \to \infty}S_V(t)=1,  $$ 
$$ S_U(t) = 1+ S_U(0) e^{-\mu_U t} \Rightarrow \lim_{t \to \infty}S_U(t)=1.   $$ 
Hence,  the first condition is satisfied. Now, to check the second condition, we first find

\[A = \left[ {\begin{array}{ccc}
\mu_H\epsilon_1+\lambda-\gamma-\mu_H & \beta_1\theta_1\alpha_1 &\kappa_1 \beta_2\theta_1\alpha_2
\\
\beta_1\theta_2 & \mu_V \epsilon_2-\mu_V & 0 
\\
\kappa_2\beta_2\theta_2 & 0 & \mu_U\epsilon_3  -\mu_U 
\end{array}}\right], \]
then, $ \hat{G}(X,Y)= AY - G(X,Y) $: 

\[\hat{G} = \left[ {\begin{array}{c}
(\alpha_1\beta_1\theta_1I_V+\kappa_1\alpha_2\beta_2\theta_1I_U+\lambda I_H)(1-S_H)
\\
\beta_1\theta_2 I_H(1-S_V)
\\
\kappa_2\beta_2\theta_2 I_H(1-S_U)
\end{array}}\right]. \]
Since $ 0 \leqslant S_H \leqslant 1 $, 		 $ 0 \leqslant S_V \leqslant 1 $ and  $ 0 \leqslant S_U \leqslant 1 $ then $\hat{G}\geqslant 0$ for all $ (X,Y)\in \Omega $. Thus,  $ Z^0 $ is globally asymptotically stable provided that $R_0\leq 1$.
\end{proof} 

%%%%%%%%%%%%%%%%%%%%%%%%%%%%%%%%%%%%%%%%%%%%%%%%%%
\subsection{Existence of Endemic Equilibrium }
The existence of endemic equilibrium is the state  where the infection cannot be totally eradicated and  the disease progration persists  in a population at all times but in relatively low frequency. Here, we discuss the existence of endemic equilibrium.

\begin{theorem}
\label{ees}
 For model (\ref{2}) there exists an endemic equilibrium $ Z^*\in\Omega$  whenever $R_0 > 1.$
 \end{theorem}

\begin{proof}

 Let $Z^*:=(S_H^*,I_H^*,R_H^*,S_V^*,I_V^*,S_r^*, I_U^*) $ be the endemic equilibrium of the model (\ref{2}) such that 

$$ S_H^*= \dfrac{\mu_H-(\mu_H+\gamma)I_H^*}{\mu_H}  , \hspace{1.0cm} R_H^* = \dfrac{\gamma I_H^*}{\mu_H},$$

$$S_V^*=\dfrac{\mu_V(1-\epsilon_2)}{\mu_V(1-\epsilon_2)+\beta_1\theta_2 I_H^*},\hspace{1.0cm} I_V^* = \dfrac{\beta_1\theta_2 I_H^*}{\mu_V(1-\epsilon_2)+\beta_1\theta_2 I_H^*}, $$

$$ S_U^*= \dfrac{\mu_U(1-\epsilon_3)}{\mu_U(1-\epsilon_3)+\kappa_2\beta_2\theta_2 I_H^*},\hspace{1.0cm}I_U^* = \dfrac{\kappa_2\beta_2\theta_2 I_H^*}{\mu_U(1-\epsilon_3)+\kappa_2\beta_2\theta_2 I_H^*},$$
\\
and $ I_H^* $ satisfies the following equation:
\\
$$   q_1 I_H^{*4} + q_2 I_H^{*3} + q_3 I_H^{*2}+q_4 I_H^* =0, $$

 where 
 $$\begin{array}{lll}
 
 q_1 & = \beta_1\beta_2\lambda\theta_2^2 \kappa_2(\mu_H+\gamma),
 \\
 &\\
  q_2& = \beta_1\beta_2\kappa_2\theta_2^2\mu_H[\gamma+\mu_H(1-\epsilon_1)][1-R_{HH}]+\beta_2\kappa_2\theta_2\mu_V
 (\mu_H+\gamma)(1-\epsilon_2)
 \\&
 [\gamma +\mu_H(1-\epsilon_1)][R_{HH}+R_{HV}]
 +\beta_1\theta_2\mu_U(1-\epsilon_3)(\mu_H+\gamma)[\gamma+\mu_H(1-\epsilon_1)][R_{HH}+R_{HU}],
 
 \\
 &\\
 q_3 & =\beta_2\kappa_2\theta_2\mu_H \mu_V(1-\epsilon_2)[\gamma+\mu_H(1-\epsilon_1)][1-R_{HH}-R_{HV}] +\beta_1\theta_2\mu_H\mu_U(1-\epsilon_3)\\&[\gamma+\mu_H(1-\epsilon_1)][1-R_{HH}-R_{HU}]
  
 +\mu_V\mu_U(1-\epsilon_2)(1-\epsilon_3)[\gamma+\mu_H(1-\epsilon_1)]\\&(\mu_H+\gamma)[R_{HH}+R_{HV}+R_{HU}],
 \\
 &\\
 q_4 & = \mu_H\mu_V\mu_U(1-\epsilon_2)(1-\epsilon_3)[\gamma+\mu_H(1-\epsilon_1)][1-(R_{HH}+R_{HV}+R_{HU})].  \end{array}$$
  \\
 Solving the above equation, we get $ I_H^*=0 $ which corresponds  to the  DFE $(Z^0 )$ and the remaining roots satisfy the cubic equation :
 \begin{equation}\label{6}
 q_1 I_H^{*3} + q_2 I_H^{*2} + q_3 I_H^{*}+q_4=0.
 \end{equation}
 
 Clearly, if $ R_{HH}+R_{HV}+R_{HU}>1$, then the above equation has a positive root since  $ q_1>0$ and $ q_4<0 $. 
Now, note that $q_3$ can be written   in terms of $q_2$ as follows:
$$q_3=p-\dfrac{\mu_H}{\mu_H+\gamma}\left( q_2-\beta_1\beta_2 \kappa_2\theta_2^2 \mu_H(\gamma+\mu_H(1-\epsilon_1)[1-R_{HH}]\right), $$
where $$\begin{array}{ll}
p&=[\gamma+\mu_H(1-\epsilon_1)]\left( \mu_H\theta_2( 
\beta_2\kappa_2 \mu_V(1-\epsilon_2)+\beta_1\mu_U(1-\epsilon_3)) \right. \\&\left. +\mu_V\mu_U(1-\epsilon_2)(1-\epsilon_3)(\mu_H+\gamma)[R_{HH}+R_{HV}+R_{HU}]\right) .
\end{array}$$
To ensure the uniqueness of the positive roots, we apply the Descartes's Sign Rule \cite{dsr}. There exits a unique positive  root when $q_2>0$ regardless the sign of $q_3$ and this happens if $R_{HH}<1$ and $R_{HH}+R_{HU}+R_{HV}>1.$ However,  when $q_2<0$ and  $R_{HH}+R_{HU}+R_{HV}>1$
there exist at least one positive root.
Note that the existence of three positive roots is only possible when $ q_2<0 $ and $q_3>0 $. %All other possibilities  $ q_2>0 $ and  $ q_3>0 $ or $ q_2<0 $ and  $ q_3<0 $ or $ q_2>0 $ and  $ q_3<0 $ are leading to have  a  unique positive root of $ I_H^* $. This happens under the conditions $ R_{HH}+R_{HV}+R_{HU}>1 $, $ R_{HH}>1, $ $ R_{HU}>1 $ and  $ R_{HH}+R_{HV}>1$ 
 %and  will be shown  numerically later.
 \end{proof}
%%%%%%%%%%%%%%%%%%%%%%%%%%%%%%%%%%%%%%%%%%%%%%%%%%%%%%

\subsection{Bifurcation Analysis}
When stability of  a system is changed as a parameter  changes  causing emergence  or disappearance of new stable points, then the system is said to undergo bifurcation. In this section, we prove that system (\ref{2}) has transcritical bifurcation. The proof is based on Sotomayor theorem described in \cite{soto}. Let $ F$ be defined as the right hand side of system (\ref{2})
 and $$ Z= (S_H,I_H,R_H,S_V,I_V,S_U,I_U)^T.$$
At $ R_0=1, $ we can check that the constant term of  characteristic equation of  $  J_{Z^0}$ is zero which implies that $  J_{Z^0}$ has a simple zero eigenvalue. Here, we choose $\lambda $ as a bifurcation parameter such that the bifurcation value corresponding to $R_0=1$ is given by
$$\lambda^*= [\mu_H(1-\epsilon_1)+\gamma][1-(R_{HV}+R_{HU  })]. $$ 
Introducing the following:
\\
 $$ a_1= -\mu_H, \quad a_2= -(\mu_H \epsilon_1+\lambda), \quad a_3=-\alpha_1\beta_1\theta_1,\quad a_4= -\kappa_1\alpha_2\beta_2\theta_1,  $$
$$b_1=-[\mu_H (1-\epsilon_1)+\lambda +\gamma], \quad b_2=   \alpha_1\beta_1\theta_1, \quad b_3=\kappa_1\alpha_2\beta_2\theta_1,\quad c_1=  \gamma,\quad c_2=  -\mu_H,  $$
$$ d_1= -\beta_1\theta_2, \quad d_2=-\mu_V, \quad d_3= -\mu _V\epsilon_2,\quad e_1= \beta_1\theta_2,\quad e_2= -\mu_V(1-\epsilon_2),   $$
$$  f_1=-\kappa_2\beta_2\theta_2 , \quad f_2=-\mu_U, \quad f_3=-\mu_U \epsilon_3,\quad g_1=\kappa_2\beta_2\theta_2,\quad g_2=-\mu_U (1-\epsilon_3),     $$
\\
and solving $ J_{(Z^0,\lambda^*)} $v $=0$, where v $ =(v_1,v_2,v_3,v_4,v_5,v_6,v_7)^T $ is a nonzero right eigenvector of $ J_{(Z^0,\lambda^*)}$ corresponding to the zero eigenvalue, we obtain 
 \\
 \[\text{v} = \left[ {\begin{array}{c}
\dfrac{1}{a_1}(-a_2+\dfrac{a_3e_1}{e_2}+\dfrac{a_4g_1}{g_2})
\\
1
\\
\dfrac{-c_1}{c_2}
\\ 
\dfrac{1}{d_2}(\dfrac{d_3e_1}{e_2}-d_1)
\\
-\dfrac{e_1}{e_2}
\\
(\dfrac{f_3g_1}{g_2}-f_1)
\\
\dfrac{-g_1}{g_2}
\end{array}}\right] v_2, \quad  v_2\neq 0. \]
Next, we find the corresponding nonzero  left eigenvector w $=(w_1,w_2,w_3,w_4,w_5,w_6,w_7)^T$ which satisfies  $ J_{(Z^0,\lambda^*)}^T$ w $=0. $ We get 
\\
\[\text{w} = \left[ {\begin{array}{c}
0
\\
1
\\
0
\\ 
0
\\
\dfrac{-b_2}{e_2}
\\
0
\\
\dfrac{-b_3}{g_2}
\end{array}}\right] w_2, \quad  w_2\neq 0. \]
\\
Model (\ref{2}) can be written as 
$dZ/dt=F(Z)$, where $F(Z)$ is the right hand side of the model. 
Now, we check the conditions of Sotomayor  theorem and begin with finding  $F_{\lambda}(\lambda^*,Z^0):$

$$ F_{\lambda}(\lambda^*,Z^0)=(0,0,0,0,0,0,0)^T. $$
So, the first condition:
\begin{equation}\label{7}
 {\text{w}}^T F_{\lambda}(\lambda^*,Z^0)=0,            
\end{equation}
 is satisfied. Next, we find the Jacobian of  $ F_{\lambda}(\lambda^*,Z)$ as follows:

\[DF_{\lambda}(\lambda^*,Z) = \left[ {\begin{array}{ccccccc}
-I_H & -S_H & 0 & 0 &0& 0 &0 
\\
I_H & S_H & 0 & 0 &0 & 0 & 0
\\
0 &0& 0 & 0 &0 & 0&0
\\
0 & 0& 0 & 0 &0 &0&0
\\ 
 0 & 0&0&0 &0 & 0&0
\\
0&0 &0&0&0&0& 0
\\
0 & 0 & 0 &0&0&0 &0
\end{array}}\right].\]

Checking the second conditions, we have
\begin{equation}\label{8}
{\text{w}}^T DF_{\lambda}(\lambda^*,Z^0){\text{v}}=w_2 v_2\neq 0.
\end{equation}
Finally, we check the third condition by finding $  D^2F(\lambda^*,Z^0) $ where $ D^2 $ denotes the matrix of the partial derivatives of each components of  $ DF(Z) $  and we get:

\[\begin{array}{l}
D^2F(\lambda^*,Z^0)({\text{v,v}})= \left[ {\begin{array}{c}
-2\lambda v_1v_2-2\alpha_1\beta_1\theta_1v_1v_3-2\kappa_1\alpha_2\beta_2\theta_1
v_1v_7
\\
2\lambda v_1v_2+2\alpha_1\beta_1\theta_1v_1v_5+2\kappa_1\alpha_2\beta_2\theta_1
v_1v_7
\\
-2\beta_1\theta_2v_4v_2
\\ 
2\beta_1\theta_2v_4v_2
\\
-2\beta_2\theta_2\kappa_2 v_6v_2
\\
2\beta_2\theta_2\kappa_2 v_6v_2
\end{array}}\right].\end{array}\]

Thus, $$\begin{array}{l}
{\text{w}}^T[ D^2F(\lambda^*,Z^0)({\text{v,v}})]=\\ \left(  2\lambda v_1v_2+2\alpha_1\beta_1\theta_2v_1v_5+2\kappa_1\alpha_2\beta_2\theta_1
v_1v_7 -2\dfrac{b_2}{e_2}\beta_1\theta_2v_4v_2 -2\dfrac{b_3}{g_2}\beta_2\theta_2\kappa_2 v_6v_2\right)  w_2.
\end{array} $$
\\ 
By substituting the values of v's, we get
\\
$$\begin{array}{ll}
{\text{w}}^T[ D^2F(\lambda^*,Z^0)({\text{v,v}})] & = 2\left[ \dfrac{1}{a_1}\left( -a_2+\dfrac{a_3e_1}{e_2} + \dfrac{a_4g_1}{g_2}\right)\left( \lambda- \alpha_1\beta_1\theta_1\dfrac{e_1}{e_2}-\kappa_1\alpha_2\beta_2\theta_1\frac{g_1}{g_2}\right) \right] w_2v_2^2 \\  
&
-2\left[ \beta_1\theta_1\left( \dfrac{d_3e_1}{e_2} -d_1\right) \dfrac{b_2}{e_2}-  \beta_2\theta_2\kappa_2 \left( \dfrac{f_3g_1}{g_2}-f_1\right) \dfrac{b_3}{g_2} \right] w_2v_2^2,  
\end{array} $$
\\
 which is nonzero since $ w_2 $ and $ v_2 $ are nonzeros.  
Hence, the system (\ref{2}) experiences a transcritical bifurcation at $ Z^0 $ as  the parameter $ \lambda $ passes through the bifurcation value $ \lambda= \lambda^*. $ The bifurcation diagram is graphed using  MATCONT package \cite{matcont} and is illustrated in Figure 
\ref{bif}.  This leads us to establish the following theorem: 
\begin{theorem}
 Model (\ref{2}) undergoes transcritical bifurcation at the DFE$(Z^0)$ when the parameter $\lambda$ passes through the bifurcation value $\lambda = \lambda^*$.
\end{theorem} 

\begin{figure}[H]
\includegraphics[height=7cm,width=17cm]{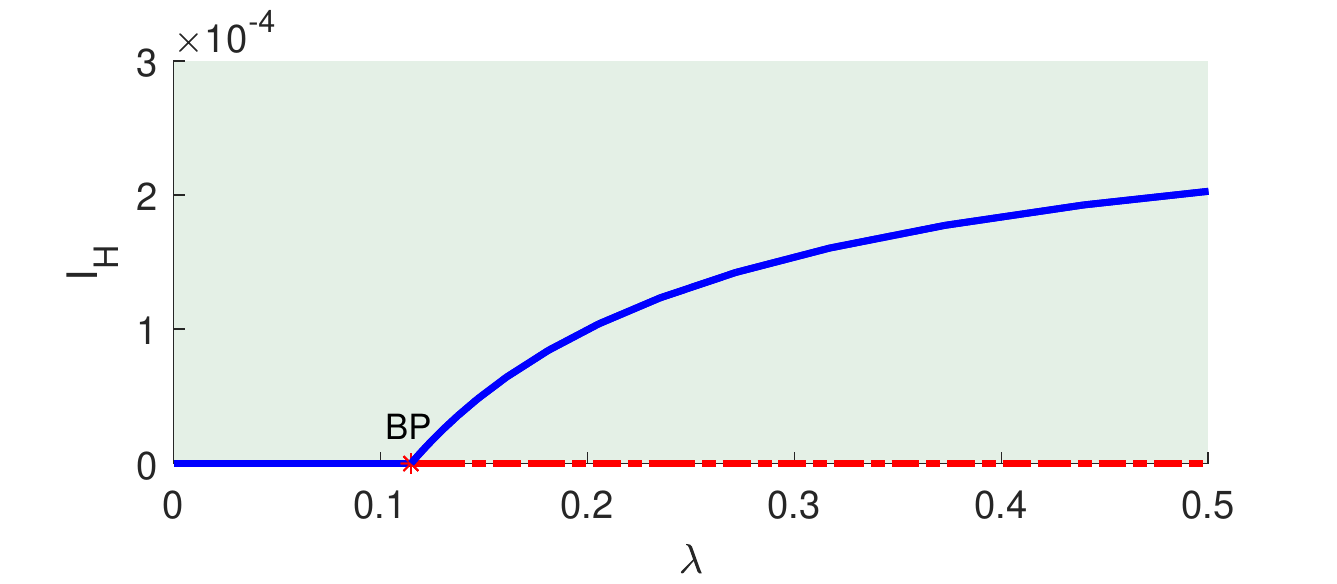}
\caption{Bifurcation figure when $\lambda$ is taken as a bifurcation parameter of system(\ref{2}) with a bifurcation value   $\lambda^* = 0.11469$  at $ R_0=1 $  and 
  by fixing parameter $ \alpha_1 = 2, \alpha_2 = 3,\kappa_1=0.5, \kappa_2 = 0.3,\epsilon_1= 0.67,\epsilon_2 = 0.06,\epsilon_3 =  0.06, \mu_H = 1/(60*365), \mu_V = 1/14, \mu_U= 1/14, \beta_1= 0.1, \beta_2= 0.15, \gamma= 0.16, \theta_1 = 0.33, \theta_2= 0.3. $ }
 \label{bif}
\end{figure} 

{\bf{Remark:}} We can establish the local stability of endemic equilibrium using the above calculations.
 We note that  based on Theorem 4 in \cite{ngm}, $ a $ and $b$ are given by
$$a=\dfrac{1}{2} {\text{w}}^T[ D^2F_{\lambda}(\lambda^*,Z^0)({\text{v,v}})]= \dfrac{1}{2}\sum_{i,j,k=1}^{n}v_i w_j w_k\dfrac{\partial^2 F_i}{\partial x_j \partial x_k }(\lambda^*,Z^0),$$ 
$$b={\text{w}}^T DF_{\lambda}(\lambda^*,Z^0){\text{v}}= \sum_{i,j=1}^{n}v_i w_j \dfrac{\partial^2 F_i}{\partial x_j \partial \lambda}(\lambda^*,Z^0).$$
According to the calculations in this section,
it is clear that  $ b \neq 0$ and $a<0$ if  $w_2$ is  positive.
Thus, there exists $ \delta > 0 $  such that the endemic equilibrium $ Z^* $ is locally asymptotically stable near $ Z^0 $ for $ 0<\lambda<\delta. $
Moreover,  according to Castillo-Chavez and
Song \cite{song} the direction of the bifurcation of the  system (\ref{2}) at $ R_0 = 1$
is forward (supercritical bifurcation).

\section{Numerical Analysis.}

In this section, the forgoing theoretical  results are confirmed by presenting the numerical results of Zika model (\ref{2}). The asymptotic behavior of the model are characterized by solving the system numerically using the baseline parameters as listed in Table \ref{table2} with appropriate initial conditions. We assume  that the human population size to be 200000, the rural mosquito population size 400000 and the forest mosquito population size 600000. These values are obtained by taken into consideration desired conditions or from literature.
The global stability (GAS) of disease  free equilibrium $ Z^0 $  is illustrated in  Figure \ref{dfe}. Moreover, the local stability of  the unique endemic equilibrium  $$ Z^* := (0.197018,0.000231, 0.802751, 0.999579,0.000420, 0.999905, 0.000095)$$  at $ q_2>0, q_3>0$
 is illustrated in Figure \ref{ee}.

\begin{figure}[H]
\subfigure{{\includegraphics[height=4.5cm,width=5.3cm]{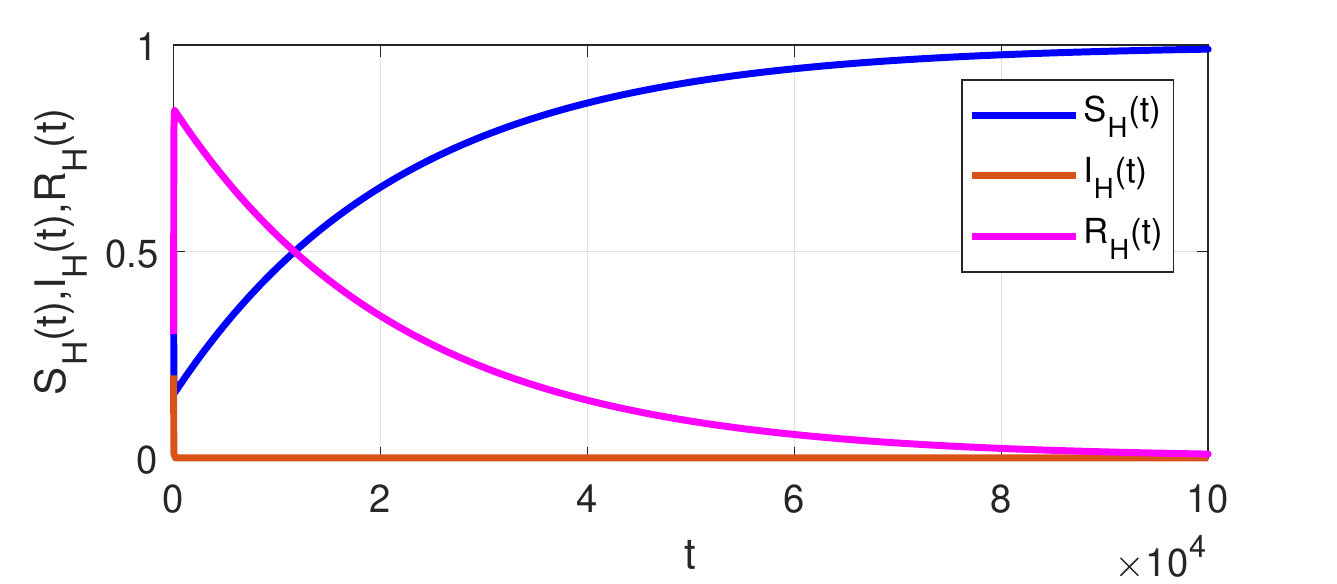}}}
\hfill
\subfigure{{\includegraphics[height=4.5cm,width=5.3cm]{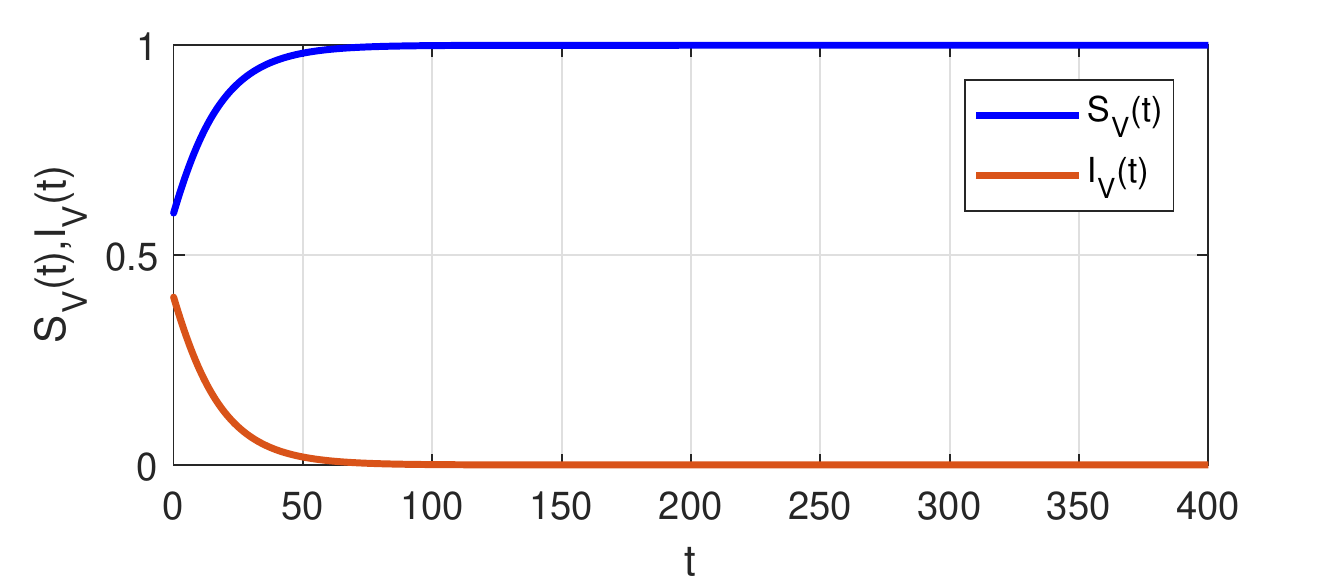}}}
\hfill
\subfigure{{\includegraphics[height=4.5cm,width=5.3cm]{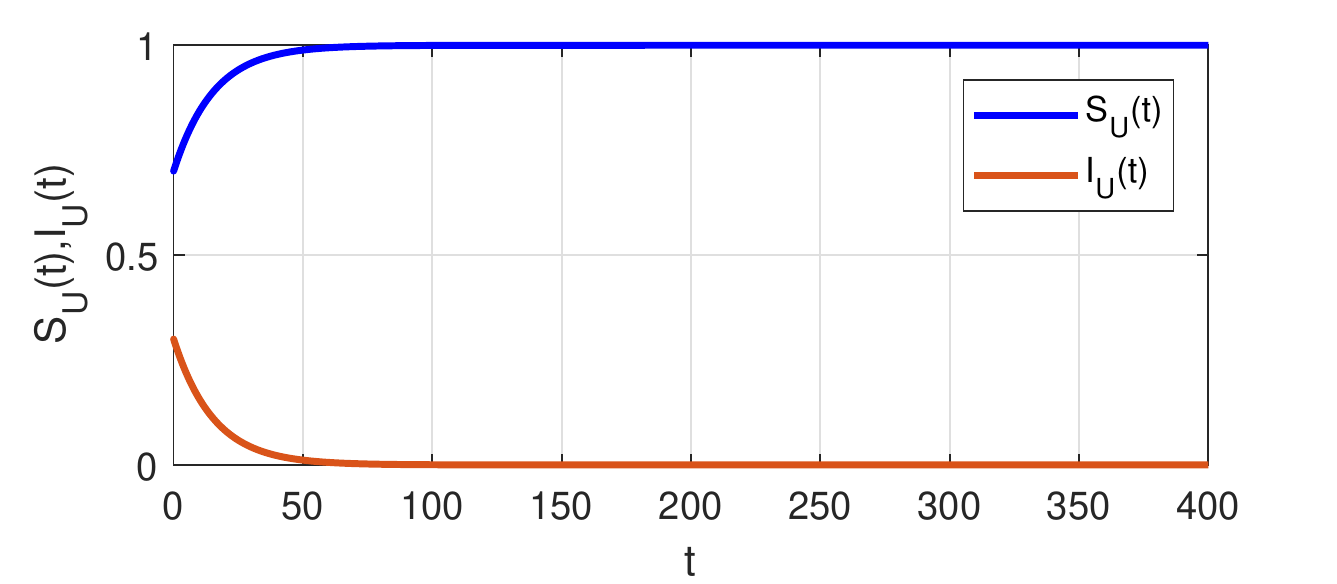}}}
\hfill
\caption{GAS of the disease free equilibrium at $ R_0= 0.87278. $ }
\label{dfe}
\end{figure}

\begin{figure}[H]
\subfigure{{\includegraphics[height=4.5cm,width=5.3cm]{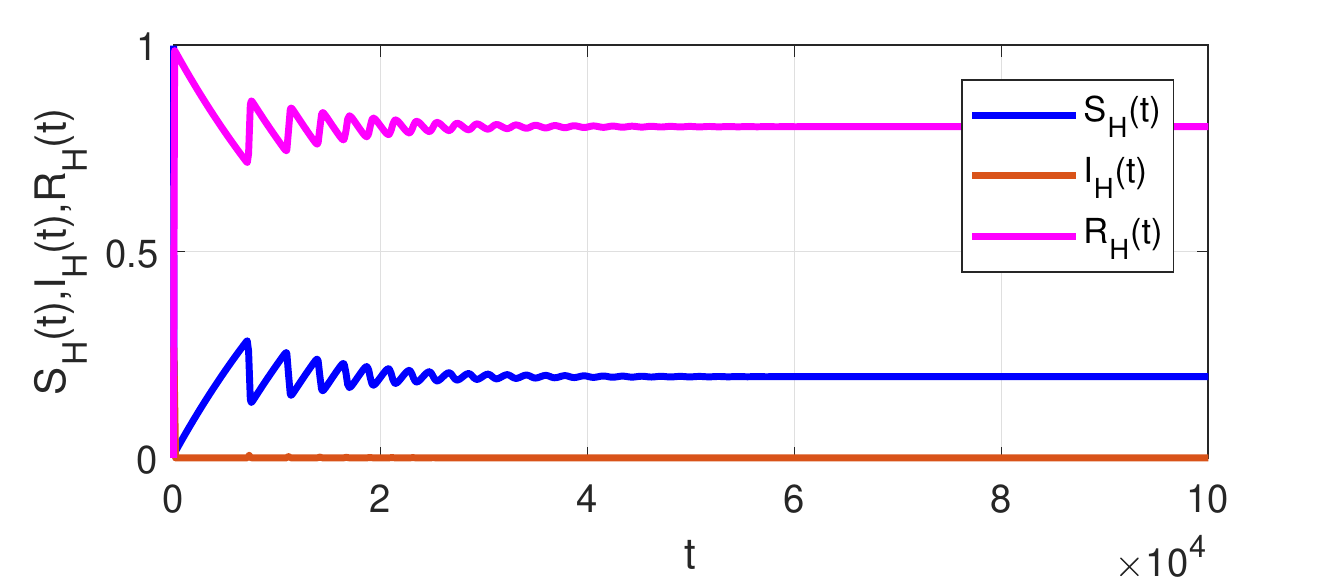}}}
\hfill
\subfigure{{\includegraphics[height=4.5cm,width=5.3cm]{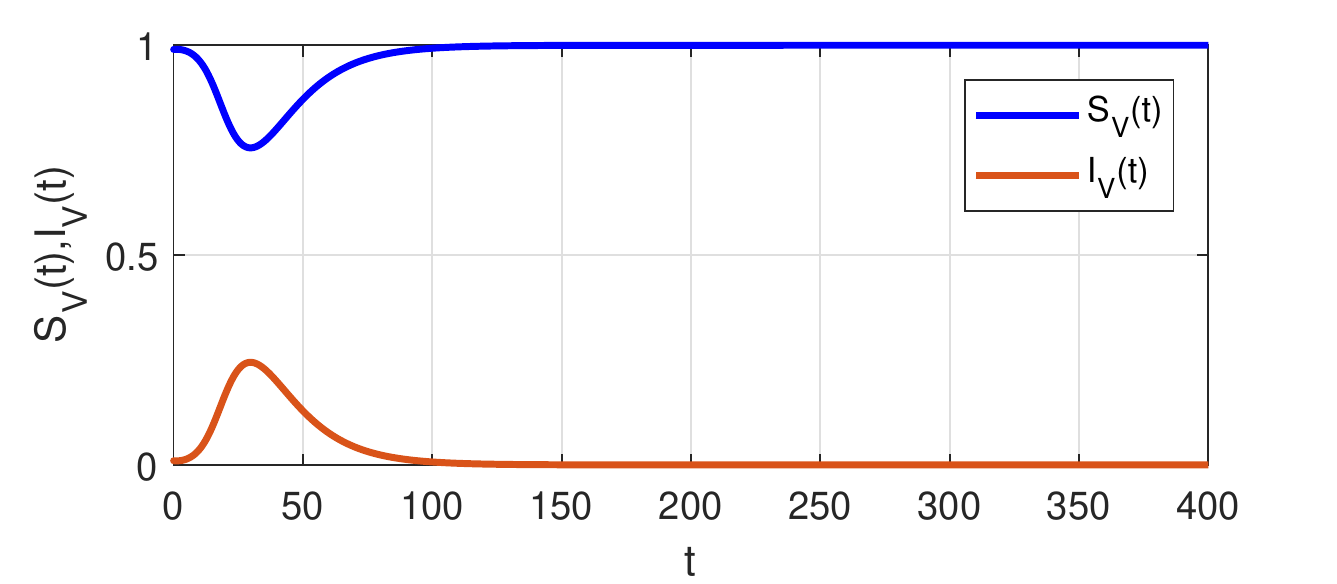}}}
\hfill
\subfigure{{\includegraphics[height=4.5cm,width=5.3cm]{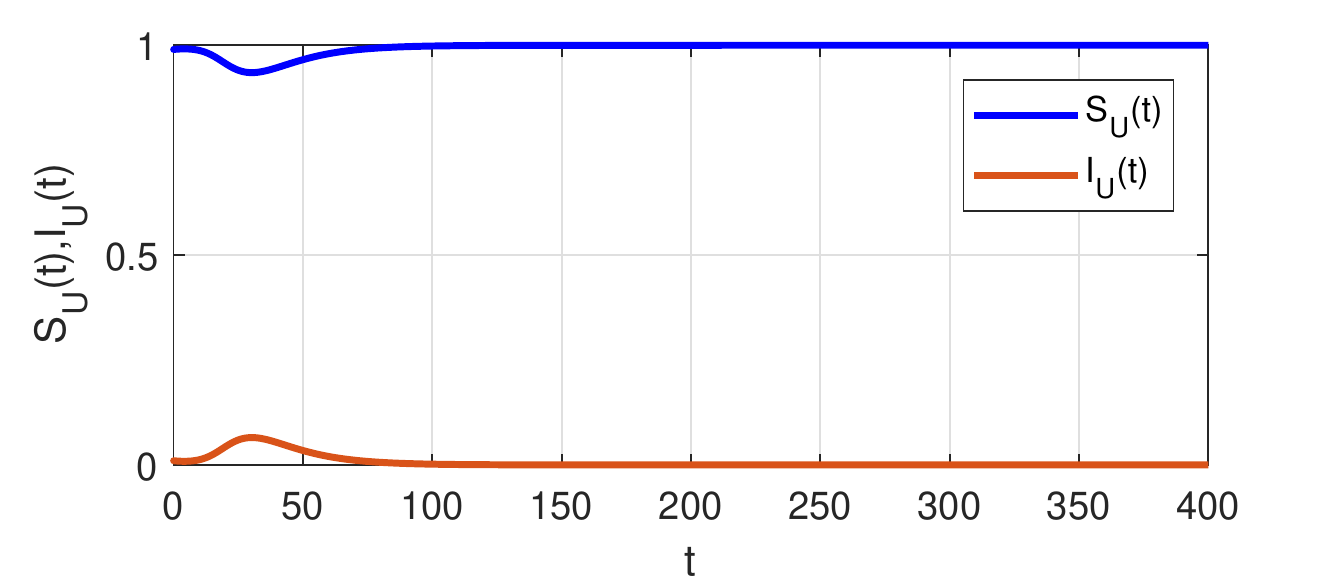}}}
\hfill
\caption{Local stability of endemic equilibrium at $ R_0= 2.8689. $}
\label{ee}
\end{figure}
%%%%%%%%%%%%%%%%%%%%%
Now, we discuss the effects  some parameter variations. In Figure   \ref{kappa_1}, we change the fraction of susceptible humans moving to forest area $\kappa_1$ and fix the other parameters. We note  that increasing the values of $\kappa_1$ leads  to  a slight increase in the maximum of both infected human and infected vector in forest areas. The infections reach their maximum and their endemic steady states  slightly more quickly as the movement of susceptible human increases. 

\begin{figure}[H]
\begin{center}
\subfigure{{\includegraphics[width=8.1cm,height=6cm]{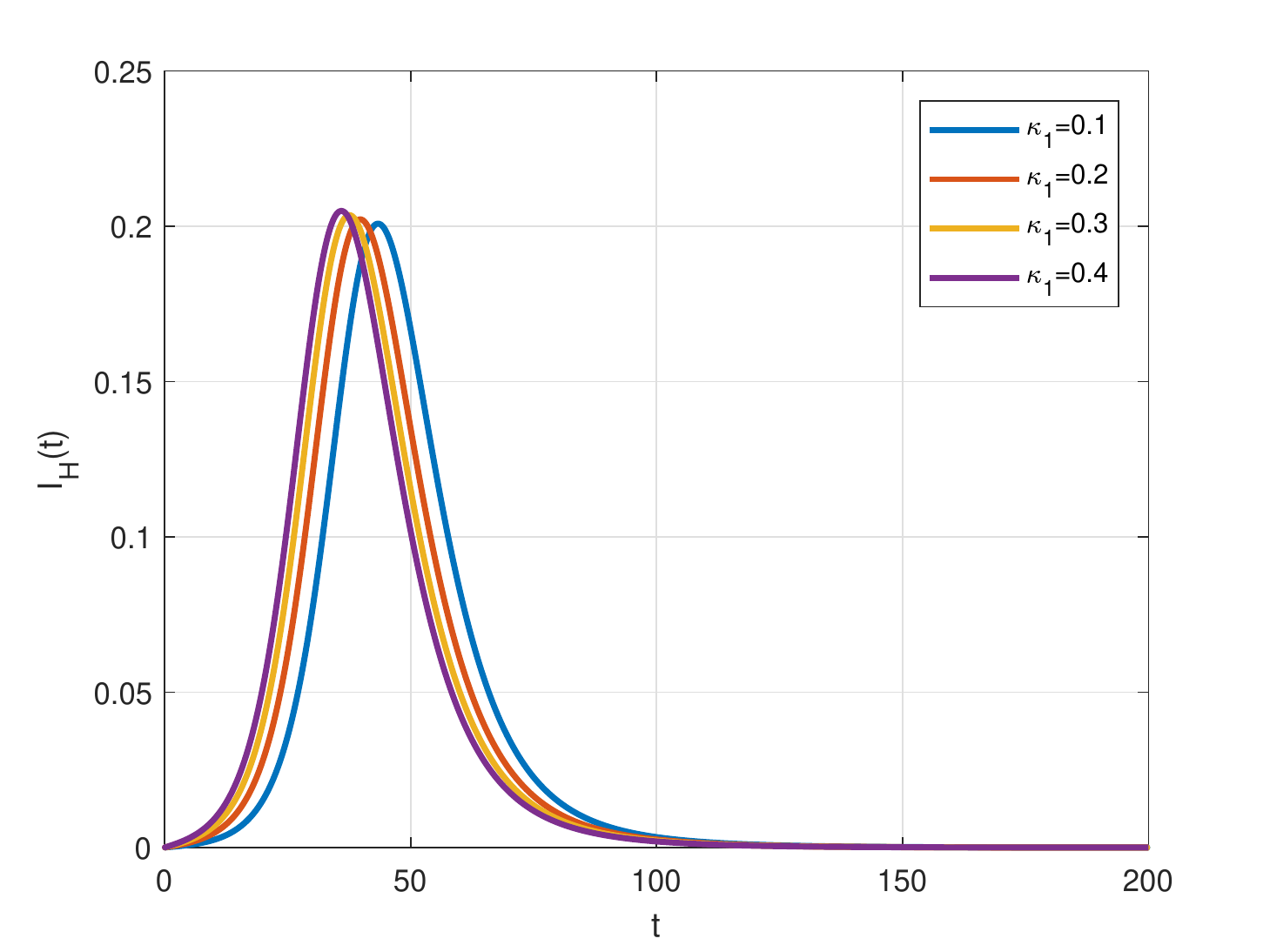}}}
\hfill
\subfigure{{\includegraphics[width=8.1cm,height=6cm]{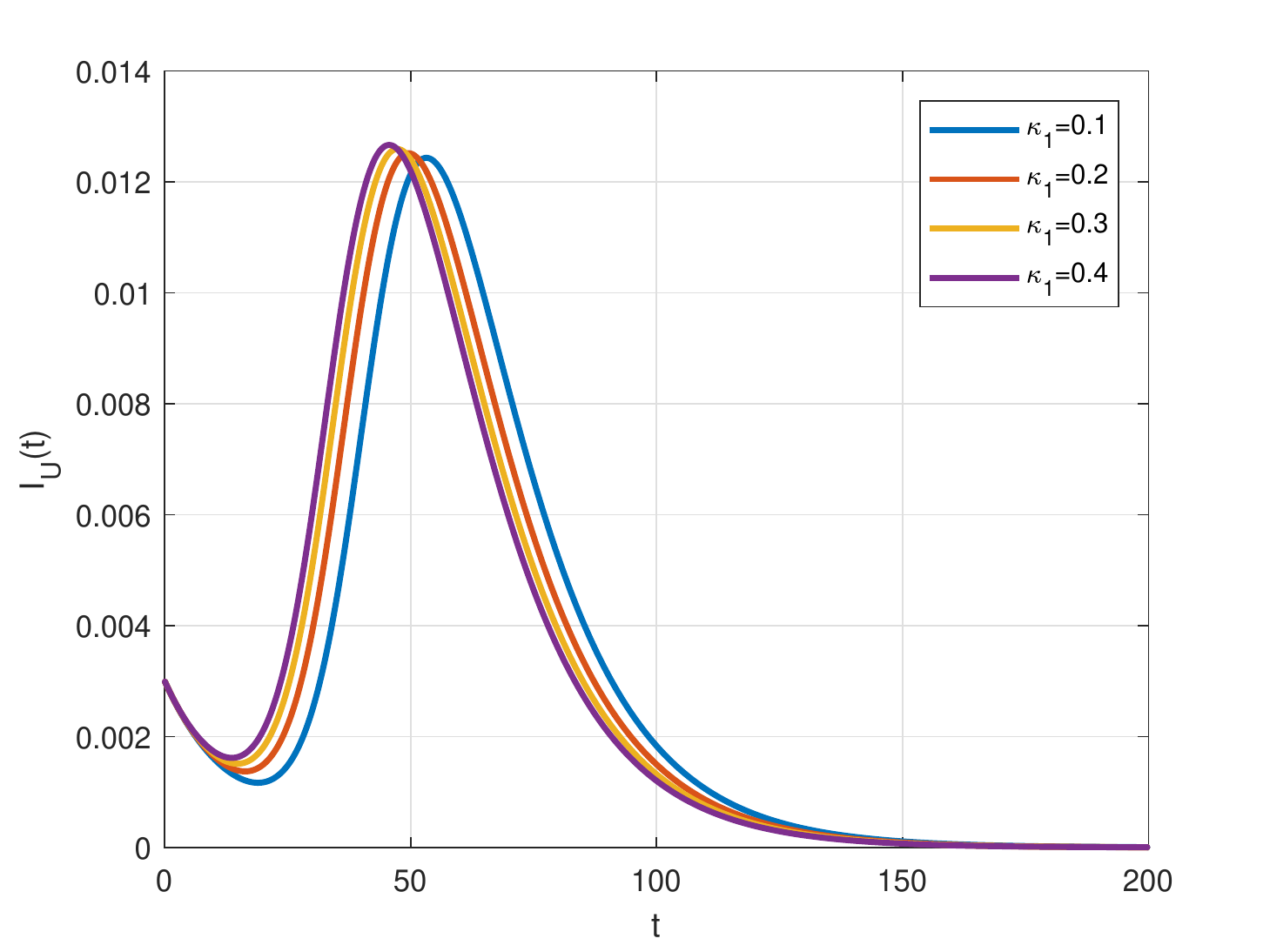}}}
\end{center}
\caption{Number of infected  populations for different values of $ \kappa_1 $ when $ \kappa_2 = 0.05. $  }
\label{kappa_1}
\end{figure}

Figure  \ref{kappa_2} illustrates the impact of varying the fraction of infected humans moving to forest area, $\kappa_2$, which shows that increasing the proportion of movement of infected humans leads to an increase in the number of infected vector in forest area and also slightly increases the infected human. Clearly, $\kappa_2$ has more impact on the infected vector compared with the infected human. This small variation in infected humans happens because the proportion of movement  of  susceptible human from rural areas to forest areas is fixed and also because the  biting rate of mosquitos is constant. Note that when $\kappa_2=0$, the  number of infected mosquitos in the forest area reaches zero which means that the disease will disappear from the forest, since the model assumes that infected humans are the only source of infection for vector population in the forest. Other sources of infection will be considered in future works.

\begin{figure}[H]
\begin{center}
\includegraphics[width=8.1cm,height=6cm]{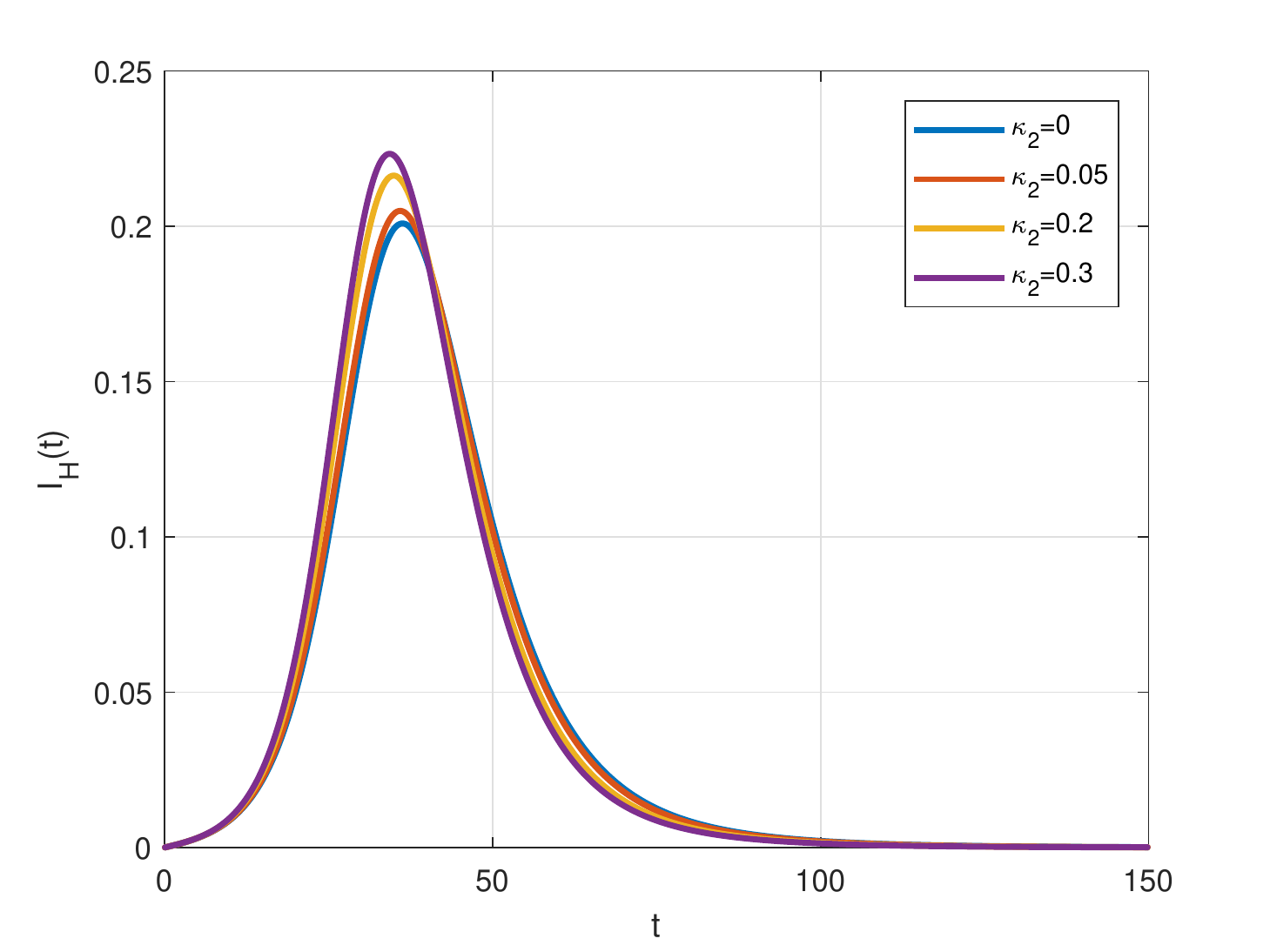}
\includegraphics[width=8.1cm,height=6cm]{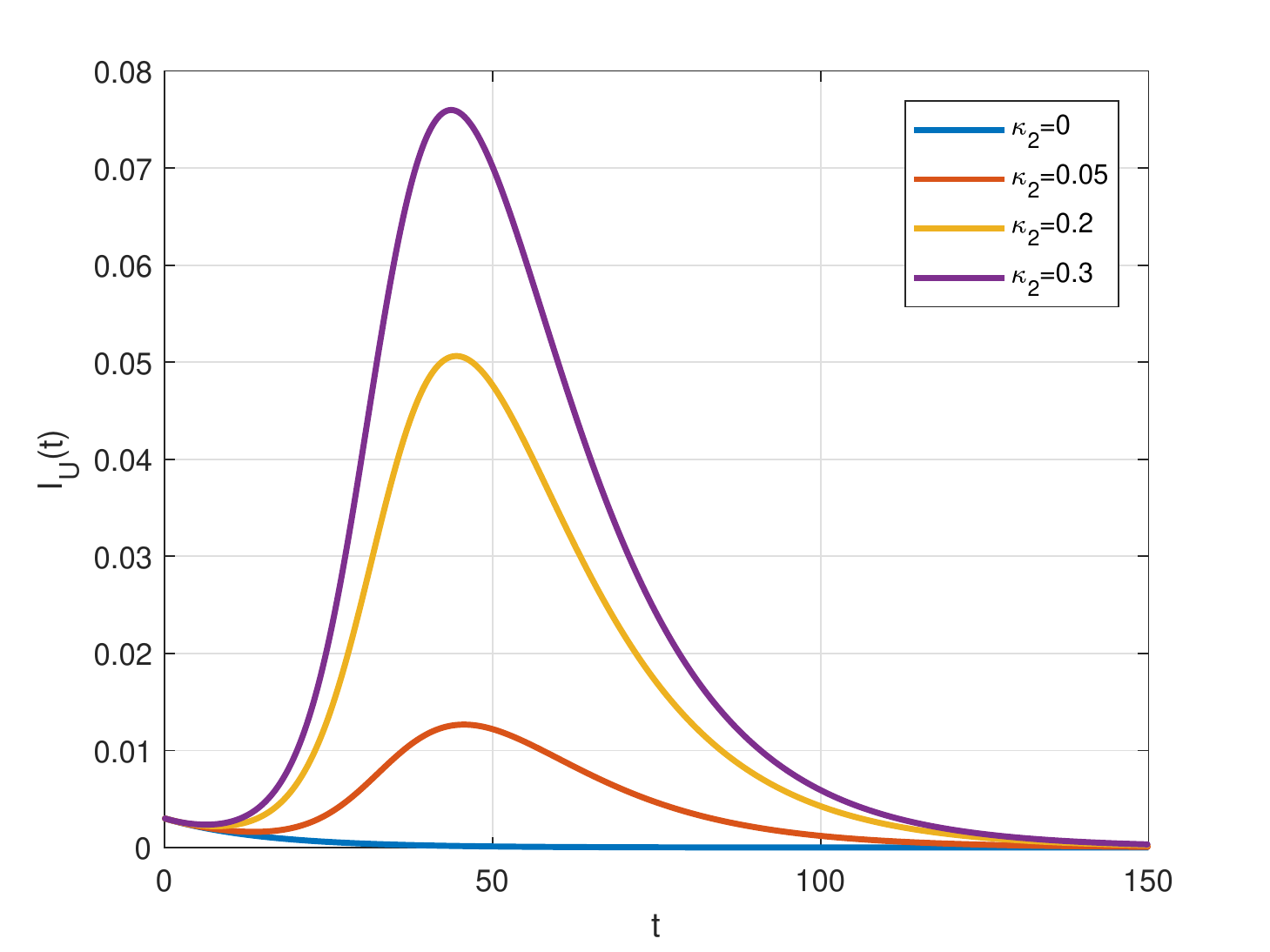}
\end{center}
\caption{Number of infected populations for different values of $ \kappa_2 $ when $ \kappa_1 = 0.4. $ }
\label{kappa_2}
\end{figure}

%%%%%%%%%%%%%%%%%%%%%%%%%%%%%%%%%%%%%%%%%%%%%%%%%%%%%%%%%%
\section{Conclusion}
A mathematical model of ZIKV disease including human movement and vertical transmission  has been proposed. The model has been analyzed and studied to investigate the impact of human movement from rural area to forest area in the spread of ZIKV. The positivity of solution and boundedness of invariant region were discussed. The basic reproduction number $R_{0}$ was computed and expressed in terms of reproduction numbers related to the interactions between humans $R_{HH},$  between  human and vector in  rural area $R_{HV}$ and between  human and vector in  forest area $R_{HU}$. It was found that the  threshold of the disease which occurs at $ R_0=1 $ is equivalent to  $R_{HH}+R_{HV}+R_{HU}=1$. Sensitivity analysis of  $R_{0}$ was carried out and it showed that $R_0$
is sensitive to almost all model parameters either positively or negatively. However, the most
positive influential parameters are the biting rate of  rural mosquitoes  on humans and  the direct  transmission rate between humans,  while the recovery rate of humans has the most negative impact. The parameters, $\kappa_1$ and $\kappa_2$, representing the proportions of susceptible and infected humans moving to forest areas, respectively, were found to  have a small positive effect on $R_0.$ Then, the local and global  stability of 
the disease free equilibrium were derived whenever $R_0$ is less than unity. Furthermore,  the system posses  endemic equilibrium and it is locally asymptotically stable when $R_0$ is greater than
unity since the direction of the  bifurcation was found to be forward. The bifurcation analysis was presented both analytically and graphically. Finally,
 numerical simulations  are presented  to demonstrate the theoretical results. The obtained figures confirmed that the human movement from rural areas to forests has a small effect in increasing the infected human and vector populations.

\bibliographystyle{plainnat}
\bibliography{Mybib}

\end{document}